\documentclass[11pt]{article}

\usepackage[T1]{fontenc}
\usepackage{lmodern}
\usepackage{microtype}
\usepackage[margin=1in]{geometry}
\usepackage{amsmath,amssymb,amsthm,mathtools}
\usepackage{booktabs,tabularx,array}
\usepackage{float}
\usepackage{enumitem}
\usepackage{xcolor}
\usepackage{url}
\usepackage{natbib}
\usepackage[hidelinks]{hyperref}
\usepackage[nameinlink,noabbrev]{cleveref}

\newcommand{\paperauthor}{Jack Fitzsimons}
\newcommand{\paperaffiliation}{Oblivious}
\newcommand{\paperemail}{jack@oblivious.com}

\newcommand{\paperkeywords}{differential privacy, zero-concentrated differential privacy, group-wise privacy, relative error, information-theoretic lower bounds}

\hypersetup{
  pdftitle={Better Privacy Guarantees for Larger Groups},
  pdfauthor={\paperauthor},
  pdfsubject={Privacy guarantees for group counts with relative expected error},
  pdfkeywords={\paperkeywords}
}

\definecolor{softgray}{gray}{0.95}
\definecolor{midgray}{gray}{0.35}

\newcolumntype{Y}{>{\raggedright\arraybackslash}X}

\newcolumntype{L}[1]{>{\raggedright\arraybackslash}p{#1}}

\newcommand{\N}{\mathcal{N}}
\newcommand{\E}{\mathbb{E}}
\newcommand{\Var}{\operatorname{Var}}
\newcommand{\KL}{D_{\mathrm{KL}}}
\newcommand{\TV}{d_{\mathrm{TV}}}
\newcommand{\pos}[1]{\left[#1\right]_{+}}
\newcommand{\mech}{\mathcal{M}}
\newcommand{\Dalpha}{D_{\alpha}}

\newtheorem{theorem}{Theorem}[section]

\newtheorem{proposition}[theorem]{Proposition}

\newtheorem{lemma}[theorem]{Lemma}

\newtheorem{corollary}[theorem]{Corollary}
\theoremstyle{definition}

\newtheorem{definition}[theorem]{Definition}

\crefname{theorem}{theorem}{theorems}
\crefname{proposition}{proposition}{propositions}
\crefname{lemma}{lemma}{lemmas}
\crefname{corollary}{corollary}{corollaries}
\crefname{definition}{definition}{definitions}
\crefname{remark}{remark}{remarks}

\setlist[itemize]{leftmargin=1.5em,itemsep=0.3em,topsep=0.4em}
\setlist[enumerate]{leftmargin=1.7em,itemsep=0.3em,topsep=0.4em}
\title{\bfseries Better Privacy Guarantees for Larger Groups}
\author{\paperauthor\\[0.25em]\small \paperaffiliation\\[-0.1em]\small \paperemail}
\date{July 15, 2026}

\begin{document}
\maketitle

\begin{abstract}
Pujol and Desfontaines asked whether a private histogram can allow more error on larger counts and use that slack to protect members of larger groups more strongly. We study this question for fixed disjoint groups under add-or-remove-one adjacency. The privacy budget $v(n)$ depends on the affected count, is nonincreasing, and must bound both R\'enyi-divergence directions at every order. This is the count-dependent form of zero-concentrated differential privacy (zCDP) studied here. The original strict relative-error condition is impossible at count zero. We therefore make the boundary tolerance explicit by requiring $\E\lvert\widehat{x}_i-x_i\rvert<r\max\{x_i,1\}$, without changing the requirement at any positive count.

Our main result determines the best dependence on group size. For the upper bound, we directly specialize an existing shifted-transformation framework. The resulting shifted-log Gaussian mechanism has a certified budget $v(n)=O_r(n^{-2})$. Conversely, for every fixed $0<r<1$, any mechanism satisfying the same positive-count utility requirement and count-dependent zCDP must have $v(n)=\Omega_r(n^{-2})$. Thus the inverse-square rate is optimal under the repaired formulation. A many-count information argument further places the leading coefficient in the large-count-then-small-error limit between $\pi/(4e^2)$ and $1/\pi$, a factor below three. At $r=1$, a data-independent release meets the repaired criterion with zero privacy loss.
\end{abstract}

\section{The problem, its boundary, and the main result}

Suppose a dataset is split into fixed, disjoint groups, and let $x_i$ be the number of records in group $i$. A standard private histogram adds noise of roughly the same size to every count. That gives every person roughly the same privacy guarantee, but it gives a large group much smaller \emph{relative} error than a small group.

\citet{PujolDesfontaines2023} asked whether one can spend this surplus accuracy differently: allow the error in group $i$ to grow in proportion to $x_i$, and use the extra noise to protect people in larger groups more strongly. Their definition asks for a privacy budget that depends only on the group count and does not increase as the count grows.

The result has two halves: a construction with an inverse-square privacy budget, and a construction-independent proof that no mechanism satisfying the same requirements can improve that rate.

\begin{theorem}[Optimal group-size rate]
\label[theorem]{thm:summary}
For every fixed $0<r<1$, the repaired formulation in \eqref{eq:repaired-utility} and \cref{def:groupzcdp} admits a mechanism with a finite nonincreasing budget
\[
  v(n)=\Theta_r(n^{-2}).
\]
Conversely, every mechanism satisfying the positive-count utility requirement and count-dependent group-wise zCDP with a finite nonincreasing budget obeys
\[
  \liminf_{n\to\infty}n^2v(n)
  \ge \frac{(1-r)^6}{128r^2(1+r)^2}>0.
\]
Hence no mechanism satisfying these requirements can make the positive-count budget decay asymptotically faster than the inverse-square rate achieved by the construction. At $r=1$, the repaired formulation admits a data-independent release with $v\equiv0$.
\end{theorem}

\begin{proof}
For the upper bound, take, for example, $r_0=r/2$ and $c=1$ in \cref{thm:main}. The converse is \cref{prop:two-count-lower}, and the endpoint is \cref{prop:r-one}.
\end{proof}

For the feasibility half, the key change of coordinates is simple.

\begin{center}
\fcolorbox{black}{softgray}{%
\begin{minipage}{0.91\textwidth}
\textbf{Upper-bound intuition.}
A one-person change is additive in count space but becomes smaller in log space:
\[
  \log(n+1+c)-\log(n+c)
  =\log\!\left(1+\frac{1}{n+c}\right)
  \asymp \frac{1}{n}.
\]
Add fixed-variance Gaussian noise in log space. Gaussian zCDP is quadratic in the change of the mean, so the privacy budget is of order $1/n^2$.
\end{minipage}}
\end{center}

The converse starts from the opposite observation. Accuracy makes nearby counts distinguishable, while privacy limits how distinguishable they can be. A two-count test determines the exponent. Using many counts at once gives the sharper constant.

\subsection{Contributions and relation to prior work}

\Cref{tab:prior} separates the inherited ingredients from the arguments developed here. The lower bounds do not assume the upper construction, and we are not aware of an earlier lower bound for this formulation.

\begin{table}[t]
\centering
\footnotesize
\begin{tabularx}{\textwidth}{@{}L{0.225\textwidth}Y Y@{}}
\toprule
\textbf{Ingredient} & \textbf{Source} & \textbf{Role in the solution} \\
\midrule
\textbf{Problem and privacy target}
& \citet{PujolDesfontaines2023} define group-wise zCDP and ask for a count-only, nonincreasing privacy budget under relative expected error.
& We keep their bidirectional, all-orders definition, prove that the literal zero-count utility condition is impossible, and state a zero-tolerant repair. \\
\addlinespace[0.35em]
\cmidrule(lr){1-3}
\addlinespace[0.25em]
\textbf{Relative and multiplicative error}
& \citet{XiaoEtAl2011} study relative-error-aware mechanisms; \citet{LeNyPappas2013} study multiplicative log-Laplace and lognormal mechanisms.
& These works motivate measuring error relatively and moving from count space to log space. \\
\addlinespace[0.35em]
\cmidrule(lr){1-3}
\addlinespace[0.25em]
\textbf{Shifted transformation mechanism}
& \citet{FinleyEtAl2026} add Gaussian noise after a concave transform, including $f(x)=\log x$ with a positive offset, building on the PRzCDP framework of \citet{SeemanEtAl2024}.
& We specialize that template to a count, choose the inverse drift for expected absolute error, and interpret the exact adjacent gap already present in their analysis as the count-local budget required here. \\
\addlinespace[0.35em]
\cmidrule(lr){1-3}
\addlinespace[0.25em]
\textbf{The $n^{-2}$ scale}
& \citet{Azize2023} studies $\N(n,r^2n^2)$ and, in a one-sided calculation, exposes an $n^{-2}$ term for this larger-groups question.
& We identify the reverse high-order divergence obstruction and use an equal-variance log representation to obtain finite zCDP in both directions at every order. \\
\addlinespace[0.35em]
\cmidrule(lr){1-3}
\addlinespace[0.25em]
\textbf{Optimality}
& \citet{BunSteinke2016} give the Gaussian zCDP calculation and the quadratic group-privacy rule.
& For every fixed $0<r<1$, we combine these tools with two-count and many-count testing arguments to prove that no mechanism satisfying the positive-count formulation can improve the $n^{-2}$ exponent. \\
\bottomrule
\end{tabularx}
\caption{Prior results and new arguments used to establish feasibility and, for every fixed $0<r<1$, the optimal positive-count rate for the repaired formulation.}
\label{tab:prior}
\end{table}

The closest upper-bound ingredient is \citet{FinleyEtAl2026}. Their template computes $f(q(D)+a)$, adds fixed-variance Gaussian noise, and applies an estimator $g$. With
\[
  q(D)=n,\qquad a=c,\qquad f(x)=\log x,
  \qquad g(y)=e^{y-\sigma^2}-c
\]
it gives the unprojected mechanism studied here. Finley et al.'s mean-unbiased estimator uses the familiar drift $-\sigma^2/2$. We use $-\sigma^2$ because it minimizes expected absolute multiplicative error.

Their analysis already contains the exact gap between adjacent transformed queries. Their stated PRzCDP policy replaces this gap by a worst case over the baseline query value. For a unit count record, that worst case is the boundary cost $\log^2((1+c)/c)/(2\sigma^2)$. Here the budget is meant to vary with the current count, so we keep the exact adjacent gap. This gives \eqref{eq:vdef}.

Related work allows privacy guarantees to depend on an individual's data or on a public policy function \citep{LuiPass2015,SeemanEtAl2024}. The relative Gaussian mechanism gives R\'enyi-DP bounds for output-dependent Gaussian variance under relative sensitivity in a broader setting \citep{HendrikxEtAl2024}. These works provide broader forms of data-dependent privacy, but they do not by themselves give the count-local feasibility and mechanism-independent converse proved here.

Relative error also appears in private hierarchical counting. \citet{BiswasEtAl2024} exploit tree structure to estimate hierarchical heavy hitters under a single global differential-privacy guarantee. That setting is complementary: the groups overlap along a hierarchy, whereas the present question concerns fixed disjoint groups and asks for the privacy budget itself to improve with the affected count.

\subsection{Scope}

The upper and lower bounds match in asymptotic order. They do not identify a pointwise-optimal budget or the exact small-error coefficient. The lower bound uses utility only at positive counts, so the choice made at zero does not affect the optimal exponent. We also follow the alternative path suggested by Pujol and Desfontaines: this paper constructs a different mechanism rather than analyzing their randomized Algorithm~1.

\section{The formal problem and the necessary zero-count repair}
\label{sec:problem}

Let a data universe $U$ be partitioned into fixed disjoint groups $U_1,\ldots,U_k$. A dataset $D$ is a finite multiset of records. Write
\[
  x(D)=(x_1(D),\ldots,x_k(D)),
\]
where $x_i(D)$ is the number of records from $U_i$. Two datasets are neighbors if one is obtained from the other by adding or removing one record.

Each neighboring edge joins two consecutive counts. We index the edge by its larger count and use one budget to bound both R\'enyi-divergence directions.

\begin{definition}[Count-dependent group-wise zCDP]
\label[definition]{def:groupzcdp}
Let $v:\mathbb{Z}_{\ge 0}\to\mathbb{R}_{\ge 0}$ be nonincreasing. For neighboring datasets $D,D'$ that differ by one record from group $i$, define
\[
  n_i(D,D')=\max\{x_i(D),x_i(D')\}.
\]
A mechanism $\mech$ satisfies count-dependent $v$-group-wise zCDP if, for every such pair and every $\alpha>1$,
\[
  \max\!\left\{
    \Dalpha\bigl(\mech(D)\,\|\,\mech(D')\bigr),
    \Dalpha\bigl(\mech(D')\,\|\,\mech(D)\bigr)
  \right\}
  \le \alpha v\bigl(n_i(D,D')\bigr).
\]
\end{definition}

This is equivalent to applying the ordered-pair condition of \citet{PujolDesfontaines2023} to both orderings of every neighboring pair. On an edge joining $n$ and $n+1$, the two orderings invoke $v(n)$ and $v(n+1)$; both divergences must therefore obey
\[
  \alpha\min\{v(n),v(n+1)\}=\alpha v(n+1),
\]
because $v$ is nonincreasing. Thus \cref{def:groupzcdp} states exactly the original two ordered-pair requirements, not a stronger substitute.

The adjective ``group-wise'' follows the terminology of the open problem and refers to histogram cells. Later, ``group privacy'' refers to the standard chaining guarantee for changing several records; the two uses are distinct. Throughout the paper, all logarithms are natural.

The utility requirement in the open problem is
\begin{equation}
  \E\lvert \widehat{x}_i-x_i\rvert < r x_i
  \qquad\text{for every group }i.
  \label{eq:literal-utility}
\end{equation}
At $x_i=0$, the right-hand side is zero.

\begin{proposition}[Why the literal statement cannot include zero]
\label[proposition]{prop:zero}
No random variable can satisfy \eqref{eq:literal-utility} when $x_i=0$. Moreover, for $0<r<1$, replacing the strict inequality by a weak one does not repair the problem if one also asks for finite two-sided R\'enyi divergence across the edge $0\leftrightarrow 1$.
\end{proposition}

\begin{proof}
At count zero, \eqref{eq:literal-utility} asks for $\E|\widehat{x}_i|<0$, which is impossible.

Now replace $<$ by $\le$. The condition $\E|\widehat{x}_i|\le 0$ forces the affected released coordinate at count zero to equal zero almost surely. Call this marginal law $P_0=\delta_0$. By data processing, finite R\'enyi divergence of the full output laws implies finite divergence of these coordinate marginals. If $\Dalpha(P_1\|P_0)$ is finite for any $\alpha>1$, then $P_1$ must be absolutely continuous with respect to $P_0$, and hence $P_1=\delta_0$ as well. Its expected error at count one is therefore $1$, which is larger than $r$ when $r<1$.
\end{proof}

There are two natural repairs. One may ask only about positive counts, or one may allow a fixed absolute tolerance at zero. We use the second form:
\begin{equation}
  \E\lvert \widehat{x}_i-x_i\rvert
  < r\max\{x_i,1\}.
  \label{eq:repaired-utility}
\end{equation}
For every nonempty group, this is exactly the original relative-error target. The value at zero merely makes the target compatible with privacy across the first adjacency edge.

\begin{proposition}[The trivial endpoint]
\label[proposition]{prop:r-one}
When $r=1$, the data-independent real-valued release $\widehat{x}_i=1/2$ satisfies \eqref{eq:repaired-utility} strictly at every count $n\ge0$ and has $v\equiv0$. Hence $0<r<1$ is the nontrivial regime.
\end{proposition}

\begin{proof}
At $n=0$, the absolute error is $1/2<1$, while at every $n\ge1$ it is $n-1/2<n$. Because the output law does not depend on the dataset, every R\'enyi divergence is zero.
\end{proof}

\section{Feasibility via a count-local shifted transformation}
\label{sec:mechanism}

Fix the requested tolerance $r\in(0,1]$. Choose two design parameters:
\[
  0<r_0<r
  \qquad\text{and}\qquad
  c>0.
\]
Choose $r_0<r$ so that the final utility bound is strict; $r_0$ may be arbitrarily close to $r$. Let $\Phi$ be the standard normal cumulative distribution function, and set
\begin{equation}
  \sigma
  =\Phi^{-1}\!\left(
      \frac12+\frac{r_0}{2(1+c)}
    \right).
  \label{eq:sigma}
\end{equation}
The argument of $\Phi^{-1}$ lies strictly between $1/2$ and $1$, so $\sigma$ is positive and finite.

For each group $i$, independently sample $Z_i\sim\N(0,1)$ and release
\begin{equation}
  \widehat{x}_i
  =\pos{
      (x_i+c)\exp(\sigma Z_i-\sigma^2)-c
    }.
  \label{eq:mechanism}
\end{equation}
Here $[y]_+=\max\{y,0\}$.

The mechanism has three steps:
\begin{enumerate}
  \item start from the shifted log-count $\log(x_i+c)$;
  \item add Gaussian noise $\sigma Z_i$ and the deterministic drift $-\sigma^2$;
  \item exponentiate, subtract $c$, and clip at zero.
\end{enumerate}
The shift $c$ makes the logarithm well-defined at zero. The drift $-\sigma^2$ is not the familiar mean-correction $-\sigma^2/2$. It is the drift that minimizes expected \emph{absolute} multiplicative error; \cref{sec:utility} makes this precise.

\begin{theorem}[Mechanism and a closed-form privacy budget]
\label[theorem]{thm:main}
The mechanism in \eqref{eq:mechanism} has the following properties.

\begin{enumerate}[label=\textup{(\roman*)}]
  \item For every group and every count $n\ge 0$,
  \[
    \E\lvert \widehat{x}_i-n\rvert
    \le \frac{r_0(n+c)}{1+c}
    < r\max\{n,1\}.
  \]

  \item It satisfies count-dependent group-wise zCDP with
  \begin{equation}
    v(0)=v(1)
    =\frac{1}{2\sigma^2}
      \log^2\!\left(\frac{1+c}{c}\right),
    \label{eq:vzero}
  \end{equation}
  and, for $n\ge 1$,
  \begin{equation}
    v(n)
    =\frac{1}{2\sigma^2}
      \log^2\!\left(\frac{n+c}{n-1+c}\right).
    \label{eq:vdef}
  \end{equation}
  The function $v$ is nonincreasing. Definition~\ref{def:groupzcdp} never invokes $v(0)$, because every neighboring edge has maximum endpoint at least $1$; setting $v(0)=v(1)$ merely defines $v$ on all nonnegative integers while preserving monotonicity.

  \item As $n\to\infty$,
  \begin{equation}
    v(n)
    =\frac{1}{2\sigma^2n^2}+O(n^{-3}).
    \label{eq:vasymptotic}
  \end{equation}
\end{enumerate}
\end{theorem}

A one-person change near count $n$ changes the log-count by about $1/n$. Gaussian zCDP cost is quadratic in this mean shift, which yields the $1/n^2$ rate. Sections~\ref{sec:utility} and~\ref{sec:privacy} verify the utility and privacy claims.

\begin{corollary}[A simple default shift]
\label[corollary]{cor:c1}
Taking $c=1$ gives
\[
  \sigma=\Phi^{-1}\!\left(\frac12+\frac{r_0}{4}\right)
\]
and
\[
  v(0)=v(1)=\frac{\log^2 2}{2\sigma^2},
  \qquad
  v(n)=\frac{\log^2(1+1/n)}{2\sigma^2}
  \quad(n\ge 1).
\]
\end{corollary}

The mechanism releases nonnegative real numbers. If integer counts are preferred, use stochastic rounding: for $y\ge0$, output $\lfloor y\rfloor$ or $\lceil y\rceil$ with probabilities chosen so that the conditional mean is $y$. This is postprocessing, so it cannot worsen privacy. Moreover, because the true count $n$ is an integer and $z\mapsto|z-n|$ is linear between consecutive integers, the conditional expected error after rounding is exactly $|y-n|$.

\section{Exact pre-clipping utility calibration}
\label{sec:utility}

The utility proof uses the following shifted-lognormal identity.

\begin{lemma}[Exact absolute multiplicative error]
\label[lemma]{lem:l1}
If $Z\sim\N(0,1)$ and $\sigma>0$, then
\begin{equation}
  \E\left|\exp(\sigma Z-\sigma^2)-1\right|
  =2\Phi(\sigma)-1.
  \label{eq:l1identity}
\end{equation}
\end{lemma}

\begin{proof}
The factor $\exp(\sigma Z-\sigma^2)$ crosses one at $Z=\sigma$. We use the Gaussian identity
\[
  e^{\sigma z-\sigma^2}\,\varphi(z)
  =e^{-\sigma^2/2}\,\varphi(z-\sigma),
\]
where $\varphi$ is the standard normal density. It gives
\[
  \E\!\left[e^{\sigma Z-\sigma^2}\mathbf{1}\{Z\le\sigma\}\right]
  =\frac12e^{-\sigma^2/2}
\]
and the same value on the event $Z>\sigma$. Splitting the absolute value at $Z=\sigma$, these two weighted terms cancel:
\begin{align*}
  \E\left|e^{\sigma Z-\sigma^2}-1\right|
  &=\E\!\left[(1-e^{\sigma Z-\sigma^2})\mathbf{1}\{Z\le\sigma\}\right]\\
  &\quad+\E\!\left[(e^{\sigma Z-\sigma^2}-1)\mathbf{1}\{Z>\sigma\}\right]\\
  &=\Phi(\sigma)-\bigl(1-\Phi(\sigma)\bigr)\\
  &=2\Phi(\sigma)-1.
\end{align*}
\end{proof}

Let
\[
  U_n=(n+c)e^{\sigma Z-\sigma^2}-c
\]
be the output before clipping. Then
\[
  U_n-n=(n+c)\left(e^{\sigma Z-\sigma^2}-1\right).
\]
By \cref{lem:l1} and the definition of $\sigma$,
\begin{align}
  \E|U_n-n|
  &=(n+c)\bigl(2\Phi(\sigma)-1\bigr)
\nonumber\\
  &=(n+c)\frac{r_0}{1+c}.
  \label{eq:utilityexact}
\end{align}
Clipping at zero cannot increase the distance to any $n\ge 0$, so \eqref{eq:utilityexact} is also an upper bound for the released value.

For $n\ge 1$,
\[
  \frac{n+c}{1+c}\le n,
\]
so the expected error is at most $r_0n<rn$. At $n=0$, it is at most $cr_0/(1+c)<r$. This proves part (i) of \cref{thm:main}.

\subsection{Why the drift is \texorpdfstring{$-\sigma^2$}{-sigma squared}}

For a fixed noise level $\sigma$, consider the more general factor $e^{\sigma Z-a}$. Dominated convergence justifies differentiating under the expectation on every compact interval of $a$. Away from a probability-zero threshold,
\begin{align*}
  \frac{d}{da}\E|e^{\sigma Z-a}-1|
  &=\E\!\left[e^{\sigma Z-a}\mathbf{1}\{Z\le a/\sigma\}\right]
    -\E\!\left[e^{\sigma Z-a}\mathbf{1}\{Z>a/\sigma\}\right]\\
  &=e^{\sigma^2/2-a}
    \left(2\Phi(a/\sigma-\sigma)-1\right).
\end{align*}
This derivative is negative for $a<\sigma^2$, zero at $a=\sigma^2$, and positive afterward. Thus $a=\sigma^2$ is the unique minimizer of expected absolute multiplicative error. By contrast, $a=\sigma^2/2$ makes the multiplicative factor have mean one. The two objectives are different because a lognormal distribution has a long right tail.

\section{Exact divergence before postprocessing}
\label{sec:privacy}

For the unprojected output $U_n$, the inequality $U_n+c>0$ holds almost surely, so
\begin{equation}
  \log(U_n+c)
  =\log(n+c)-\sigma^2+\sigma Z.
  \label{eq:lognormal}
\end{equation}
Thus, after the bijective transformation $u\mapsto\log(u+c)$, the output law is the Gaussian
\[
  \N\bigl(\log(n+c)-\sigma^2,\,\sigma^2\bigr).
\]
Because the variance is independent of $n$, every adjacent pair has finite R\'enyi divergence at all orders.

For equal-variance Gaussians, the R\'enyi divergence is exact \citep[Lemma 2.4]{BunSteinke2016}:
\begin{equation}
  \Dalpha\bigl(\N(\mu,\sigma^2)\,\|\,\N(\nu,\sigma^2)\bigr)
  =\frac{\alpha(\mu-\nu)^2}{2\sigma^2}.
  \label{eq:gaussianrdp}
\end{equation}
R\'enyi divergence is invariant under a measurable bijection and obeys data processing \citep{VanErvenHarremoes2014}. Applying \eqref{eq:gaussianrdp} to \eqref{eq:lognormal} gives, for any counts $n,m\ge 0$,
\begin{equation}
  \Dalpha(U_n\|U_m)
  =\Dalpha(U_m\|U_n)
  =\frac{\alpha}{2\sigma^2}
    \log^2\!\left(\frac{n+c}{m+c}\right).
  \label{eq:allpairrdp}
\end{equation}
Clipping and optional stochastic rounding are postprocessing, so the same expression remains a valid upper bound for the released estimate.

For the adjacent unprojected edge $n-1\leftrightarrow n$, equation \eqref{eq:allpairrdp} is exactly $\alpha v(n)$ with $v(n)$ from \eqref{eq:vdef}, in both directions. After clipping and optional rounding, data processing gives the corresponding upper bound. The boundary edge $0\leftrightarrow1$ uses $v(1)$, because its larger endpoint is $1$; the chosen extension $v(0)=v(1)$ is not invoked. Thus the released mechanism satisfies \cref{def:groupzcdp}.

The monotonicity is immediate from
\[
  \log\!\left(\frac{n+c}{n-1+c}\right)
  =\log\!\left(1+\frac{1}{n-1+c}\right),
\]
which decreases with $n$. Finally, a Taylor expansion gives
\[
  \log\!\left(1+\frac{1}{n-1+c}\right)
  =\frac{1}{n}+O(n^{-2}),
\]
and hence \eqref{eq:vasymptotic}.

For the full histogram, neighboring datasets change only one coordinate. All other coordinate laws are identical, and independent product factors with identical laws contribute zero R\'enyi divergence. Therefore the product argument proves part (ii) of \cref{thm:main}; part (iii) follows from the scalar formula \eqref{eq:vdef} and its expansion above.

\subsection{Interpreting the rate}

A zCDP guarantee $\rho$ implies
\[
  \bigl(\rho+2\sqrt{\rho\log(1/\delta)},\delta\bigr)\text{-DP}
\]
for every $\delta>0$ \citep{BunSteinke2016}. Thus, for the mechanism above, $v(n)=\Theta(n^{-2})$ corresponds, for fixed $\delta$, to an ordinary approximate-DP epsilon of order $1/n$ in the regime where the square-root term dominates.

The parameter $c$ controls a boundary tradeoff. Smaller $c$ permits a larger $\sigma$ for the same error and therefore improves the large-$n$ constant. But it makes the jump from $0$ to $1$ larger in log space, so $v(0)=v(1)$ worsens. The choice $c=1$ in \cref{cor:c1} is a simple default rather than a universal optimum.

Two limit regimes make this tradeoff precise. First fix $c>0$, let $n\to\infty$, and then let $r\downarrow0$ while $r_0/r\to1$. Then
\[
  \sigma\sim \frac{r}{1+c}\sqrt{\frac{\pi}{2}},
  \qquad
  \lim_{r\downarrow0}\lim_{n\to\infty}r^2n^2v(n)
  =\frac{(1+c)^2}{\pi},
\]
whereas the boundary cost satisfies
\[
  v(1)\sim
  \frac{(1+c)^2}{\pi r^2}
  \log^2\!\left(\frac{1+c}{c}\right).
\]
Thus $c=1$ has large-count coefficient $4/\pi$. Taking $c\downarrow0$ approaches $1/\pi$, but the first-edge budget diverges. For example, the diagonal choice $c=r$ used below has $v(1)$ of order $r^{-2}\log^2(1/r)$.

Second fix $0<r<1$, let $r_0\uparrow r$, and let $c\downarrow0$. The leading large-count coefficient approaches
\[
  \frac{1}{2\,\Phi^{-1}((1+r)/2)^2}.
\]
For small $r$,
\[
  \Phi^{-1}((1+r)/2)
  =r\sqrt{\frac{\pi}{2}}+O(r^3),
\]
so the large-count budget approaches
\begin{equation}
  v(n)\sim \frac{1}{\pi r^2n^2}.
  \label{eq:smallrconstant}
\end{equation}
This is an iterated limit: first $n\to\infty$, then $r_0\uparrow r$ and $c\downarrow0$, and finally $r\downarrow0$. The coefficient is approached within the allowed family. No fixed choice with $r_0<r$ and $c>0$ need attain it, and we do not claim that it is exactly optimal.

\section{Mechanism-independent lower bounds: two counts, then many}
\label{sec:lower}

The upper construction alone does not show whether the inverse-square rate is intrinsic. The lower bounds below show that it is. They apply to arbitrary mechanisms, including mechanisms with correlated coordinates, multiple noise stages, or output laws without densities. For a vector release, postprocessing to one coordinate preserves the utility assumption and cannot increase the privacy loss.

We state utility in this section with the weak inequality $\E_m|Y-m|\le rm$. This enlarges the admissible class, so every lower bound below also applies to the manuscript's strict positive-count criterion.

Testing two counts proves the exponent for every fixed $0<r<1$. Testing many counts also uses accuracy at the intermediate values and yields a stronger small-error constant.

\begin{center}
\fcolorbox{black}{softgray}{%
\begin{minipage}{0.91\textwidth}
\textbf{The one-line lower-bound idea.}
Choose a hidden offset $J$, so the true count is $n+J$. Accuracy says the release must reveal information about $J$. Privacy says it cannot reveal too much. Comparing the two statements forces a lower bound on $v(n)$.
\end{minipage}}
\end{center}

For a discrete random variable $J$, write $H(J)$ for its Shannon entropy.

The following pathwise lemma replaces the varying edge budgets along a fixed chain by their common upper bound.

\begin{lemma}[Pathwise zCDP group privacy]
\label[lemma]{lem:pathwise-group}
Let $D_0,\ldots,D_k$ be a fixed chain of datasets such that consecutive datasets are neighbors. Suppose that every edge of this chain satisfies $\rho$-zCDP in both directions. If $P_j$ is the output law on $D_j$, then, for every $\alpha>1$,
\[
  \max\{\Dalpha(P_0\|P_k),\Dalpha(P_k\|P_0)\}
  \le \alpha k^2\rho.
\]
\end{lemma}

\begin{proof}
This is the standard quadratic group-privacy theorem for zCDP \citep[Proposition 1.9]{BunSteinke2016}, applied to the fixed chain $D_0,\ldots,D_k$. Its proof uses only the zCDP inequalities for consecutive laws on that chain, so no global uniform bound outside the path is required.
\end{proof}

\subsection{A two-count rate bound}

The two-count argument proves the exponent for every $0<r<1$. The later many-count argument improves the constant as $r\downarrow0$.

\begin{proposition}[Rate lower bound for $0<r<1$]
\label[proposition]{prop:two-count-lower}
Fix $0<r<1$. Write $\E_m$ for expectation when the affected-group count is $m$. Suppose a mechanism produces an estimate $Y$ with
\[
  \E_m|Y-m|\le rm
  \qquad\text{for every positive count }m,
\]
and satisfies count-dependent group-wise zCDP for a finite, nonincreasing function $v$. Define
\[
  C_r=\left(\frac{1+3r}{1-r}\right)^2.
\]
Then, for every $n\ge1$,
\begin{equation}
  v(n)
  \ge
  \frac{(1-r)^2}
       {2\bigl(\lceil C_rn\rceil-n\bigr)^2}.
  \label{eq:two-count-finite}
\end{equation}
Consequently,
\begin{equation}
  \liminf_{n\to\infty}n^2v(n)
  \ge
  \frac{(1-r)^6}{128r^2(1+r)^2},
  \label{eq:two-count-asymptotic}
\end{equation}
so $v(n)=\Omega_r(n^{-2})$.
\end{proposition}

\begin{proof}
Fix $n$ and put
\[
  m=\lceil C_rn\rceil,
  \qquad
  q=\sqrt{m/n},
  \qquad
  t=\sqrt{mn}=nq.
\]
Let $P_n$ and $P_m$ be the output laws on nested datasets with affected-group counts $n$ and $m$. Utility first separates these laws; privacy will then limit that separation. The finite expected absolute-error bounds make the relevant deviations integrable, so Markov's inequality gives
\[
  P_n(Y>t)\le\frac{rn}{t-n},
  \qquad
  P_m(Y\le t)\le\frac{rm}{m-t}.
\]
Since $m=nq^2$, the sum of these two bounds is
\[
  r\frac{q+1}{q-1}.
\]
The definition of $C_r$ ensures $q\ge(1+3r)/(1-r)$, so this sum is at most $(1+r)/2$. For $A=\{Y\le t\}$,
\begin{equation}
  \TV(P_n,P_m)
  \ge P_n(A)-P_m(A)
  \ge\frac{1-r}{2}.
  \label{eq:two-count-tv}
\end{equation}

Every one-person edge on the path from count $n$ to count $m$ has budget at most $v(n)$. Applying \cref{lem:pathwise-group} with $\rho=v(n)$ therefore gives, for every $\alpha>1$,
\[
  \Dalpha(P_n\|P_m)\le\alpha(m-n)^2v(n).
\]
Using the limit of R\'enyi divergence to KL divergence as $\alpha\downarrow1$ \citep{VanErvenHarremoes2014} gives
\[
  \KL(P_n\|P_m)\le(m-n)^2v(n).
\]
Pinsker's inequality and \eqref{eq:two-count-tv} give
\[
  \KL(P_n\|P_m)
  \ge2\TV(P_n,P_m)^2
  \ge\frac{(1-r)^2}{2}.
\]
Combining the last two displays proves \eqref{eq:two-count-finite}. Finally,
\[
  C_r-1=\frac{8r(1+r)}{(1-r)^2},
\]
which gives \eqref{eq:two-count-asymptotic} after multiplying by $n^2$ and taking $n\to\infty$.
\end{proof}

\subsection{Absolute error forces information}

The following lemma converts expected absolute estimation error into mutual information. Its $1/4$ term is the mean absolute value of an independent uniform dither on $[-1/2,1/2]$.

\begin{lemma}[Absolute error forces information]
\label[lemma]{lem:error-information}
Let $J$ be an integer-valued random variable with finite support, let $Y$ take values in a standard Borel space, and let $\widetilde J(Y)$ be any real-valued estimator. If
\[
  D=\E\lvert J-\widetilde J(Y)\rvert,
\]
then
\begin{equation}
  I(J;Y)
  \ge
  H(J)-\log\!\left(2e\left(D+\frac14\right)\right).
  \label{eq:error-information}
\end{equation}
\end{lemma}

\begin{proof}
Let $U$ be independent of $(J,Y)$ and uniform on $[-1/2,1/2]$, and set $X=J+U$. The unit intervals around distinct integers overlap only at endpoints. Thus $h(X)=H(J)$. Conditional on $Y=y$, the law of $X$ is a mixture of uniform densities on those disjoint intervals, with weights $\Pr(J=j\mid Y=y)$. Each component has differential entropy zero because its interval has length one. Therefore
\[
  h(X\mid Y=y)=H(J\mid Y=y)
\]
for almost every $y$, and hence $I(X;Y)=I(J;Y)$ \citep{CoverThomas2006}.

Use the same estimator for $X$. By the triangle inequality,
\[
  \E\lvert X-\widetilde J(Y)\rvert
  \le D+\E|U|
  =D+\frac14.
\]
A real random variable $W$ with $\E|W|\le d$ has differential entropy at most $\log(2ed)$; the centered Laplace law is the maximizer \citep{CoverThomas2006}. This follows by comparing the density of $W$ with $(2d)^{-1}e^{-|w|/d}$ and using nonnegativity of Kullback-Leibler divergence.

Because $Y$ is standard Borel, regular conditional laws exist. For almost every $y$, define
\[
  d(y)=\E\!\left[
    \left|X-\widetilde J(y)\right|
    \,\middle|\,Y=y
  \right].
\]
The entropy bound gives $h(X\mid Y=y)\le\log(2e\,d(y))$. Averaging over $y$ and applying Jensen's inequality gives
\[
  h(X\mid Y)
  \le
  \log\!\left(2e\,\E\lvert X-\widetilde J(Y)\rvert\right).
\]
Hence
\[
  I(J;Y)
  = I(X;Y)
  =h(X)-h(X\mid Y),
\]
which is exactly \eqref{eq:error-information}.
\end{proof}

\subsection{The general many-count inequality}

\begin{theorem}[Many-count inequality for privacy and utility]
\label[theorem]{thm:manycount}
Fix $0<r<1$. Write $\E_m$ for expectation when the affected-group count is $m$. Suppose a mechanism produces an estimate $Y$ such that, at every positive count $m$,
\begin{equation}
  \E_m|Y-m|\le rm.
  \label{eq:many-utility}
\end{equation}
Suppose also that it satisfies count-dependent group-wise zCDP for a finite, nonincreasing function $v$.

Fix a base count $n\ge1$. Let $J$ be any finitely supported random variable on the nonnegative integers, and let $j_\star$ be any nonnegative integer with $\E(J-j_\star)^2>0$. Then
\begin{equation}
  v(n)
  \ge
  \frac{
    \pos{
      H(J)-
      \log\!\left(
        2e\left(r(n+\E J)+\frac14\right)
      \right)
    }
  }{
    \E(J-j_\star)^2
  }.
  \label{eq:general-many-count}
\end{equation}
\end{theorem}

\begin{proof}
Fix a nested sequence of datasets $D_0,D_1,\ldots$ in which $D_j$ has affected-group count $n+j$ and each step adds one record from that group. Choose the hidden offset $J$, run the mechanism on $D_J$, and call the relevant released estimate $Y$. Let $P_j$ be its law on $D_j$, and use $P_{j_\star}$ as a reference law.

The proof has two parts. Privacy upper-bounds $I(J;Y)$; utility lower-bounds it.

Every count on the path from $n+j$ to $n+j_\star$ is at least $n$. Since $v$ is nonincreasing, every one-person edge on that path has zCDP budget at most $v(n)$. With $k=|j-j_\star|$, \cref{lem:pathwise-group} gives, for every $\alpha>1$,
\[
  \Dalpha(P_j\|P_{j_\star})
  \le \alpha k^2v(n).
\]
Using the limit of R\'enyi divergence to KL divergence as $\alpha\downarrow1$ \citep{VanErvenHarremoes2014} gives
\begin{equation}
  \KL(P_j\|P_{j_\star})
  \le
  (j-j_\star)^2v(n).
  \label{eq:path-kl}
\end{equation}
Let $\overline P=\sum_j\Pr(J=j)P_j$ be the mixture law of $Y$. Averaging \eqref{eq:path-kl} over $J$ and using the standard mutual-information/relative-entropy identity \citep{CoverThomas2006} gives, for every reference law $Q$,
\[
  \E_J\KL(P_J\|Q)
  =I(J;Y)+\KL(\overline P\|Q).
\]
Taking $Q=P_{j_\star}$ and dropping the final nonnegative term yields
\begin{equation}
  I(J;Y)
  \le
  v(n)\E(J-j_\star)^2.
  \label{eq:information-upper}
\end{equation}

Now use utility. With the estimator $\widetilde J(Y)=Y-n$,
\begin{align*}
  \E\lvert J-(Y-n)\rvert
  &=\E\lvert(n+J)-Y\rvert\\
  &\le r\E(n+J)\\
  &=r(n+\E J).
\end{align*}
By \cref{lem:error-information},
\begin{equation}
  I(J;Y)
  \ge
  H(J)-
  \log\!\left(
    2e\left(r(n+\E J)+\frac14\right)
  \right).
  \label{eq:information-lower}
\end{equation}
Combining \eqref{eq:information-upper} and \eqref{eq:information-lower} gives the stated ratio when its numerator is positive. If the numerator is negative, the bound is vacuous; since $v(n)\ge0$, the positive part covers both cases.
\end{proof}

Equation~\eqref{eq:general-many-count} displays three competing quantities. The entropy $H(J)$ strengthens the information lower bound. The mean $\E J$ enlarges the allowed error, and $\E(J-j_\star)^2$ is the privacy cost. A strong prior balances all three.

\subsection{A fully explicit finite bound}

The uniform distribution gives a clean formula with no limiting argument.

\begin{corollary}[Uniform many-count bound]
\label[corollary]{cor:uniform-lower}
Under the assumptions of \cref{thm:manycount}, let $K\ge3$ be odd. Then, for every $n\ge1$,
\begin{equation}
  v(n)
  \ge
  \frac{12}{K^2-1}
  \pos{
    \log
    \frac{K}{
      2e\left(
        r\left(n+\frac{K-1}{2}\right)+\frac14
      \right)
    }
  }.
  \label{eq:uniform-finite}
\end{equation}
\end{corollary}

\begin{proof}
Take $J$ uniform on $\{0,1,\ldots,K-1\}$ and set $j_\star=(K-1)/2$. Then
\[
  H(J)=\log K,
  \qquad
  \E J=\frac{K-1}{2},
  \qquad
  \E(J-j_\star)^2=\frac{K^2-1}{12}.
\]
Substitute these three identities into \eqref{eq:general-many-count}.
\end{proof}

Fix $r\in(0,1)$ and $a>0$, and choose odd integers $K_n\ge3$ such that $K_n/(rn)\to a$. Then
\begin{equation}
  \liminf_{n\to\infty}r^2n^2v(n)
  \ge
  \frac{12}{a^2}
  \pos{
    \log\frac{a}{2e(1+ar/2)}
  }.
  \label{eq:uniform-asymptotic}
\end{equation}
If we next let $r\downarrow0$, the right side is maximized at $a=2e^{3/2}$ and equals
\begin{equation}
  \frac{3}{2e^3}
  \approx0.0746806.
  \label{eq:uniform-constant}
\end{equation}
The uniform-prior bound improves the two-count constant by a factor of $192/e^3\approx9.56$.

\subsection{A bell-shaped prior gives the sharper constant}

The next construction uses a bell-shaped prior to increase entropy relative to squared spread. The following quantization lemma transfers this continuous-prior advantage to integer counts.

\begin{lemma}[Fine integer quantization]
\label[lemma]{lem:quantization}
Let $Z$ have a bounded, compactly supported, piecewise-continuous density $f$. Let $Q_S$ be $SZ$ rounded to the nearest integer, with either convention at ties. As $S\to\infty$,
\begin{align}
  H(Q_S)&=\log S+h(Z)+o(1),
  \label{eq:quant-entropy}\\
  \E Q_S&=S\E Z+O(1),
  \label{eq:quant-mean}\\
  \E Q_S^2&=S^2\E Z^2+O(S).
  \label{eq:quant-second}
\end{align}
\end{lemma}

\begin{proof}
For each quantization cell $C_{S,k}$, let
\[
  f_S(z)=S\int_{C_{S,k}}f(u)\,du
  \qquad(z\in C_{S,k}).
\]
If $p_{S,k}=\Pr(Q_S=k)$, then $p_{S,k}=f_S(z)/S$ on that cell and hence
\[
  H(Q_S)-\log S
  =-\int f_S(z)\log f_S(z)\,dz.
\]
Piecewise continuity gives $f_S\to f$ almost everywhere. The functions $f_S$ are uniformly bounded and supported in a common compact enlargement of the support of $f$, so dominated convergence applied to $t\mapsto t\log t$ proves \eqref{eq:quant-entropy}. The pointwise bound $|Q_S-SZ|\le1/2$ gives \eqref{eq:quant-mean}. Writing $Q_S=SZ+\Delta_S$ with $|\Delta_S|\le1/2$, boundedness of $Z$ gives $2SZ\Delta_S+\Delta_S^2=O(S)$ pointwise, which proves \eqref{eq:quant-second}.
\end{proof}

\begin{corollary}[Small-error lower-bound constant]
\label[corollary]{cor:smallr-lower}
For each $r\in(0,1)$, let a mechanism satisfy the assumptions of \cref{thm:manycount}, with budget $v_r$. Then
\begin{equation}
  \liminf_{r\downarrow0}
  \liminf_{n\to\infty}
  r^2n^2v_r(n)
  \ge
  \frac{\pi}{4e^2}
  \approx0.1062921.
  \label{eq:gaussian-constant}
\end{equation}
\end{corollary}

\begin{proof}
We take the limits in the order displayed in the result: first $n\to\infty$, then $r\downarrow0$, and finally the truncation level $L\to\infty$.

Fix $L>0$. Let $Z_L$ be a standard Gaussian conditioned on $[-L,L]$, and write
\[
  h_L=h(Z_L),
  \qquad
  \tau_L^2=\Var(Z_L).
\]
The distribution is symmetric, so $\E Z_L=0$.

Fix $x>0$ and put $S=xrn$. Let $Q_S$ be $SZ_L$ rounded to the nearest integer, let
\[
  B=\lceil LS\rceil+1,
  \qquad
  J=B+Q_S,
  \qquad
  j_\star=B.
\]
The shift $B$ makes $J$ nonnegative. By \cref{lem:quantization}, as $n\to\infty$ with $r,L,x$ fixed,
\begin{align*}
  H(J)&=\log S+h_L+o(1),\\
  \E(J-B)^2&=S^2\tau_L^2+O(S),\\
  \E J&=LS+o(S).
\end{align*}
Apply \cref{thm:manycount} and use $S=xrn$. Here $\log S=\log(xrn)$, while the utility term inside the logarithm is
\[
  2e\left(r(n+\E J)+\frac14\right)
  =2e\,rn\bigl(1+Lxr+o(1)\bigr).
\]
Thus the $\log(rn)$ terms cancel. After multiplying by $r^2n^2$ and taking $n\to\infty$,
\begin{equation}
  \liminf_{n\to\infty}r^2n^2v_r(n)
  \ge
  \frac{
    \pos{
      \log x+h_L-\log(2e)-\log(1+Lxr)
    }
  }{
    x^2\tau_L^2
  }.
  \label{eq:truncated-lower}
\end{equation}
Now let $r\downarrow0$. For fixed $L$, optimize the right side over $x$. The optimum satisfies
\[
  \log x+h_L-\log(2e)=\frac12.
\]
Equivalently, $x=x_L=2e^{3/2-h_L}$, and the positive-part bracket equals $1/2$. This gives
\begin{equation}
  \liminf_{r\downarrow0}
  \liminf_{n\to\infty}r^2n^2v_r(n)
  \ge
  \frac{e^{2h_L}}{8e^3\tau_L^2}.
  \label{eq:truncated-constant}
\end{equation}
Finally let $L\to\infty$. The truncated Gaussian converges to the standard Gaussian in entropy and variance:
\[
  h_L\to\frac12\log(2\pi e),
  \qquad
  \tau_L^2\to1.
\]
Equation \eqref{eq:truncated-constant} therefore tends to
\[
  \frac{2\pi e}{8e^3}
  =\frac{\pi}{4e^2}.
\]
\end{proof}

The Gaussian is not an arbitrary guess. Within the continuous limiting-prior and scale-optimization argument above, a limiting prior $Z$ produces the factor
\[
  \frac{e^{2h(Z)}}{8e^3\Var(Z)}.
\]
A Gaussian maximizes differential entropy at fixed variance \citep{CoverThomas2006}, so this quantity is at most $\pi/(4e^2)$. The truncated Gaussian sequence above reaches that limit. Thus Corollary~\ref{cor:smallr-lower} is optimal within this continuous-prior template; this observation does not rule out discrete, multiscale, or $r$-dependent uses of \cref{thm:manycount}.

Call a family $\{(\mech_r,v_r):0<r<1\}$ admissible if, for every $r\in(0,1)$, the mechanism $\mech_r$ meets the positive-count utility condition $\E_m|Y-m|\le rm$ for all $m\ge1$ and has a finite, nonincreasing count-dependent group-wise zCDP budget $v_r$. For each such family, form the iterated leading coefficient
\[
  \liminf_{r\downarrow0}\liminf_{n\to\infty}
  r^2n^2v_r(n).
\]
Let $C_\star$ be the infimum of this quantity over all admissible families. The upper construction and Corollary~\ref{cor:smallr-lower} give
\begin{equation}
  \frac{\pi}{4e^2}
  \;\le\;
  C_\star
  \;\le\;
  \frac1\pi.
  \label{eq:constant-sandwich}
\end{equation}
For the upper bound, choose $r_0(r)=r(1-r)$ and $c(r)=r$ for $0<r<1$. This gives an admissible family. Since $r_0(r)/r\to1$ and $c(r)\to0$, \eqref{eq:smallrconstant} gives the coefficient $1/\pi$. The ratio between the two displayed constants is
\[
  \frac{1/\pi}{\pi/(4e^2)}
  =\frac{4e^2}{\pi^2}
  \approx2.995.
\]
Thus the lower and upper bounds are now within a factor of three.

\section{Why the unequal-variance Gaussian fails the definition}
\label{sec:hetero}

A tempting mechanism releases
\begin{equation}
  X_n\sim\N(n,r^2n^2).
  \label{eq:hetero}
\end{equation}
It has the right scale of expected error. \citet{Azize2023} studies this deterministic-variance mechanism and derives a one-sided term proportional to $\alpha/(r^2n^2)$, revealing the quadratic scale. The analysis below applies only to this mechanism, not to the randomized-variance Algorithm~1 of \citet{PujolDesfontaines2023}. Unequal neighboring variances create a tail obstruction: one R\'enyi-divergence direction becomes infinite at sufficiently high orders.

\begin{lemma}[Finite R\'enyi divergence for unequal Gaussians]
\label[lemma]{lem:unequal-gaussian}
Let $P=\N(\mu_P,s_P^2)$ and $Q=\N(\mu_Q,s_Q^2)$ with positive variances. For $\alpha>1$, $\Dalpha(P\|Q)$ is finite if and only if
\begin{equation}
  \alpha s_Q^2+(1-\alpha)s_P^2>0.
  \label{eq:gaussian-finite}
\end{equation}
\end{lemma}

\begin{proof}
The defining integral for R\'enyi divergence is $\int p^\alpha q^{1-\alpha}$. Its exponent is a quadratic polynomial in the integration variable. The coefficient of the negative quadratic term is
\[
  \frac12\left(\frac{\alpha}{s_P^2}
               +\frac{1-\alpha}{s_Q^2}\right).
\]
The integral is finite exactly when this coefficient is positive. Multiplying by $s_P^2s_Q^2$ gives \eqref{eq:gaussian-finite}. At equality the remaining exponential is nonintegrable as well.
\end{proof}

The boundary edge must be handled separately because $X_0=\N(0,0)=\delta_0$, whereas \cref{lem:unequal-gaussian} assumes positive variances. Since the nondegenerate law $X_1=\N(1,r^2)$ is not absolutely continuous with respect to $X_0$,
\[
  \Dalpha(X_1\|X_0)=\infty
  \qquad\text{for every }\alpha>1.
\]
Thus the edge $0\leftrightarrow1$ fails the required bound at every R\'enyi order.

For a positive edge $n\leftrightarrow n+1$ with $n\ge1$, apply the lemma to the reverse divergence from the larger-variance law at count $n+1$ to the smaller-variance law at count $n$. It is finite only if
\[
  \alpha n^2+(1-\alpha)(n+1)^2>0,
\]
or equivalently
\begin{equation}
  \alpha<\frac{(n+1)^2}{2n+1}.
  \label{eq:alpha-threshold}
\end{equation}
Consequently every positive edge has infinite reverse divergence at and above this finite threshold. The group-wise zCDP definition requires both directions for \emph{every} $\alpha>1$, so \eqref{eq:hetero} fails at the boundary for every order and on every positive edge above a finite order.

Azize's calculation identifies the correct $n^{-2}$ scale, but an all-orders zCDP construction must avoid unequal neighboring variances. The shifted log-Gaussian mechanism does so: although its count-space variance changes with the count, every transformed law has variance $\sigma^2$. The analysis of Pujol and Desfontaines's Algorithm~1 remains open; the mechanism above proves feasibility without resolving it.

\section{The remaining coefficient question}
\label{sec:discussion}

The inverse-square rate is now established for $0<r<1$. One quantitative question remains: the exact value of the iterated coefficient $C_\star$ in \eqref{eq:constant-sandwich}. This section places the known bounds side by side and explains where further improvement would have to come from.

\begin{table}[H]
\centering
\small
\begin{tabularx}{\textwidth}{@{}Y L{0.17\textwidth} L{0.25\textwidth}@{}}
\toprule
\textbf{Argument} & \textbf{Role} & \textbf{Small-$r$ coefficient} \\
\midrule
Two-count proof using Markov and Pinsker & lower bound & $C_\star\ge 1/128\approx0.00781$ \\
Uniform many-count prior & lower bound & $C_\star\ge 3/(2e^3)\approx0.07468$ \\
Truncated Gaussian prior & lower bound & $C_\star\ge \pi/(4e^2)\approx0.10629$ \\
Shifted log-Gaussian mechanism & upper bound & $C_\star\le 1/\pi\approx0.31831$ \\
\bottomrule
\end{tabularx}
\caption{Known small-$r$ leading coefficients for the remaining coefficient question.}
\label{tab:constants}
\end{table}

\paragraph{The zero-count repair is part of the statement.}
The original strict relative-error requirement is logically impossible at zero. Equation \eqref{eq:repaired-utility} is a natural repair, not a consequence of the original wording. Readers interested only in nonempty groups may instead impose the original requirement for $n\ge1$; the same mechanism and proofs apply.

\paragraph{Why the two-count proof stops at $1/128$.}
Write $m-n=arn$ and place a threshold at $t=n+\theta(m-n)$. In the small-$r$ limit, Markov's inequality and Pinsker's inequality give at best
\[
  \frac{2}{a^2}
  \pos{
    1-\frac{1}{a\theta(1-\theta)}
  }^2.
\]
This expression is maximized by $\theta=1/2$ and $a=8$, giving exactly $1/128$. Thus the threshold choices are optimal within the template using Markov's inequality, one threshold, and Pinsker's inequality. Proposition~\ref{prop:two-count-lower} keeps this argument because it proves the rate for every $0<r<1$. The large constant improvement in Corollary~\ref{cor:smallr-lower} comes from changing the experiment: it uses many counts simultaneously.

\paragraph{zCDP does not make the released count sub-Gaussian.}
The sub-Gaussian random variable in zCDP is the privacy loss, not the numerical output \citep{BunSteinke2016}. A finite-support output distribution can have finite R\'enyi divergence at every order as long as neighboring laws have common support and bounded likelihood ratios. Therefore one cannot simply replace Markov's inequality by a zCDP tail bound for $|Y-n|$. The useful coupling between privacy and utility is informational, as in Theorem~\ref{thm:manycount}.

\paragraph{The remaining factor is below three.}
For small $r$, the shifted log-Gaussian mechanism approaches the upper constant $1/\pi\approx0.31831$, while Corollary~\ref{cor:smallr-lower} gives $\pi/(4e^2)\approx0.10629$. The ratio is about $2.995$. After invoking zCDP's quadratic group-privacy theorem, the information-radius step retains only the resulting KL endpoint bound. A sharper lower bound may need to exploit the higher R\'enyi orders more directly.

\paragraph{The lognormal right tail is not the leading-order cost.}
For the asymptotically optimized choice $r_0/r\to1$ and $c\to0$, small $r$ gives $\sigma=r\sqrt{\pi/2}\,(1+o(1))$ and
\[
  e^{\sigma Z-\sigma^2}-1
  =\sigma Z+O_{L^1}(\sigma^2).
\]
So, to first order, the shifted log-Gaussian mechanism is simply additive Gaussian noise in count space. The skew appears only in higher-order terms. This suggests that merely correcting the visible right tail will not change the leading constant. A genuine improvement, if one exists, should change the local noise experiment. Recent work on exact zCDP characterizations of non-Gaussian mechanisms and direct optimization of noise laws under R\'enyi objectives provides tools for exploring that direction \citep{HarrisonManurangsi2025,GilaniEtAl2025}.

\paragraph{The utility criterion is expected absolute error.}
The open problem asks for an expectation, and the mechanism is calibrated exactly for that objective. High-probability relative error, squared error, bias, and confidence intervals are different design criteria. The drift and the constants would change.

\paragraph{The groups are fixed and disjoint.}
The proof uses the setting in the original problem: adding or removing one record changes one histogram coordinate. Overlapping or data-dependent groups require a separate sensitivity analysis.

\paragraph{The mechanism is easy to compose.}
Because the guarantee is zCDP, repeated mechanisms compose by adding their edge budgets pointwise: if release $t$ has budget $v_t$, then $v_{\mathrm{tot}}(n)=\sum_t v_t(n)$. The same statement holds for adaptive releases when each conditional mechanism satisfies its claimed edgewise bound uniformly over previous transcripts. This is one practical advantage of obtaining pure all-orders zCDP rather than a bound that is valid only up to a count-dependent R\'enyi order.

\paragraph{The mechanism may be biased.}
The drift $-\sigma^2$ is chosen to minimize expected absolute error, not to make the estimate unbiased. For a counting release, bias can be reported or corrected only at the cost of changing the absolute-error calculation. The theorem makes no unbiasedness claim.

\section{Conclusion}

The original strict relative-error condition is impossible at count zero. After making that boundary tolerance explicit, the larger-groups question has a sharp answer at the level of asymptotic rate.

Fix $0<r<1$. Any mechanism satisfying the positive-count utility requirement and count-dependent zCDP has $v(n)=\Omega_r(n^{-2})$. The shifted log-Gaussian mechanism matches this bound with $v(n)=\Theta_r(n^{-2})$. The converse is mechanism-independent, so the inverse-square rate is not an artifact of the construction or of the zero-count repair. At $r=1$, a data-independent release gives zero privacy loss.

The upper construction is a count-local specialization of the transformation framework of \citet{FinleyEtAl2026}. Their framework supplies the shifted transform and the exact transformed-query gap. We choose the inverse drift for expected absolute error, retain the count-specific gap, and verify bidirectional all-orders zCDP after postprocessing. The equal-variance log representation avoids the obstruction faced by the natural unequal-variance Gaussian, whose reverse divergence is infinite at high R\'enyi orders.

The many-count argument narrows the remaining coefficient gap to
\[
  \frac{\pi}{4e^2}\le C_\star\le\frac1\pi,
\]
and these bounds differ by less than a factor of three. Determining the exact value of $C_\star$ is the main question left open.

\paragraph{Acknowledgments.}
This note was prompted by the open problem of David Pujol and Damien Desfontaines. The author is grateful for the transformation-mechanism framework developed by Finley et al.\ and for Achraf Azize's public 2023 calculations. The many-count lower bound grew out of asking where a two-count proof loses information.

\begingroup
\raggedright
\bibliographystyle{plainnat}
\bibliography{references}

@misc{PujolDesfontaines2023,
  author       = {David Pujol and Damien Desfontaines},
  title        = {Open problem - Better privacy guarantees for larger groups},
  howpublished = {\url{https://differentialprivacy.org/open-problem-better-privacy-guarantees-for-larger-groups/}},
  year         = {2023},
  month        = jun,
  note         = {Accessed July 14, 2026}
}

@misc{Azize2023,
  author       = {Achraf Azize},
  title        = {Open problem - Better privacy guarantees for larger groups},
  howpublished = {Technical note, \url{https://achraf-azize.github.io/Open_problem_Better_privacy_guarantees_for_larger_groups.pdf}},
  year         = {2023},
  month        = jul,
  note         = {Accessed July 14, 2026}
}

@inproceedings{LeNyPappas2013,
  author    = {{Le Ny}, J\'er\^ome and George J. Pappas},
  title     = {Privacy-Preserving Release of Aggregate Dynamic Models},
  booktitle = {Proceedings of the 2nd ACM International Conference on High Confidence Networked Systems},
  series    = {HiCoNS '13},
  pages     = {49--56},
  publisher = {Association for Computing Machinery},
  address   = {New York, NY, USA},
  year      = {2013},
  doi       = {10.1145/2461446.2461454}
}

@inproceedings{XiaoEtAl2011,
  author    = {Xiaokui Xiao and Gabriel Bender and Michael Hay and Johannes Gehrke},
  title     = {{iReduct}: Differential Privacy with Reduced Relative Errors},
  booktitle = {Proceedings of the 2011 ACM SIGMOD International Conference on Management of Data},
  pages     = {229--240},
  publisher = {Association for Computing Machinery},
  year      = {2011},
  doi       = {10.1145/1989323.1989348}
}

@inproceedings{BunSteinke2016,
  author    = {Mark Bun and Thomas Steinke},
  title     = {Concentrated Differential Privacy: Simplifications, Extensions, and Lower Bounds},
  booktitle = {Theory of Cryptography},
  editor    = {Martin Hirt and Adam Smith},
  series    = {Lecture Notes in Computer Science},
  volume    = {9985},
  pages     = {635--658},
  publisher = {Springer},
  year      = {2016},
  doi       = {10.1007/978-3-662-53641-4_24},
  eprint    = {1605.02065},
  archivePrefix = {arXiv}
}

@inproceedings{LuiPass2015,
  author    = {Edward Lui and Rafael Pass},
  title     = {Outlier Privacy},
  booktitle = {Theory of Cryptography},
  editor    = {Yevgeniy Dodis and Jesper Buus Nielsen},
  series    = {Lecture Notes in Computer Science},
  volume    = {9015},
  pages     = {277--305},
  publisher = {Springer},
  year      = {2015},
  doi       = {10.1007/978-3-662-46497-7_11}
}

@inproceedings{HendrikxEtAl2024,
  author    = {Hadrien Hendrikx and Paul Mangold and Aur\'elien Bellet},
  title     = {The Relative Gaussian Mechanism and its Application to Private Gradient Descent},
  booktitle = {Proceedings of the 27th International Conference on Artificial Intelligence and Statistics},
  editor    = {Sanjoy Dasgupta and Stephan Mandt and Yingzhen Li},
  series    = {Proceedings of Machine Learning Research},
  volume    = {238},
  pages     = {3079--3087},
  publisher = {PMLR},
  year      = {2024},
  url       = {https://proceedings.mlr.press/v238/hendrikx24a.html}
}

@article{BiswasEtAl2024,
  author  = {Ari Biswas and Graham Cormode and Yaron Kanza and Divesh Srivastava and Zhengyi Zhou},
  title   = {Differentially Private Hierarchical Heavy Hitters},
  journal = {Proceedings of the ACM on Management of Data},
  volume  = {2},
  number  = {5},
  pages   = {208:1--208:25},
  year    = {2024},
  month   = nov,
  doi     = {10.1145/3695826}
}

@article{SeemanEtAl2024,
  author  = {Jeremy Seeman and William Sexton and David Pujol and Ashwin Machanavajjhala},
  title   = {Privately Answering Queries on Skewed Data via Per-Record Differential Privacy},
  journal = {Proceedings of the VLDB Endowment},
  volume  = {17},
  number  = {11},
  pages   = {3138--3150},
  year    = {2024},
  doi     = {10.14778/3681954.3681989}
}

@article{FinleyEtAl2026,
  author  = {Brian Finley and Anthony M. Caruso and Justin C. Doty and Ashwin Machanavajjhala and Mikaela R. Meyer and David Pujol and William Sexton and Zachary Terner},
  title   = {Slowly Scaling Per-Record Differential Privacy},
  journal = {Journal of Privacy and Confidentiality},
  volume  = {16},
  number  = {1},
  year    = {2026},
  month   = apr,
  doi     = {10.29012/jpc.992}
}

@book{CoverThomas2006,
  author    = {Thomas M. Cover and Joy A. Thomas},
  title     = {Elements of Information Theory},
  edition   = {2nd},
  publisher = {John Wiley \& Sons},
  address   = {Hoboken, NJ},
  year      = {2006},
  doi       = {10.1002/047174882X}
}

@article{VanErvenHarremoes2014,
  author  = {Tim van Erven and Peter Harremo{\o}s},
  title   = {{R}\'enyi Divergence and {K}ullback--{L}eibler Divergence},
  journal = {{IEEE} Transactions on Information Theory},
  volume  = {60},
  number  = {7},
  pages   = {3797--3820},
  year    = {2014},
  month   = jul,
  doi     = {10.1109/TIT.2014.2320500}
}

@inproceedings{GilaniEtAl2025,
  author    = {Atefeh Gilani and Juan Felipe Gomez and Shahab Asoodeh and Flavio P. Calmon and Oliver Kosut and Lalitha Sankar},
  title     = {Optimizing Noise Distributions for Differential Privacy},
  booktitle = {Proceedings of the 42nd International Conference on Machine Learning},
  series    = {Proceedings of Machine Learning Research},
  volume    = {267},
  pages     = {19505--19522},
  publisher = {PMLR},
  year      = {2025},
  url       = {https://proceedings.mlr.press/v267/gilani25a.html}
}

@misc{HarrisonManurangsi2025,
  author        = {Charlie Harrison and Pasin Manurangsi},
  title         = {Exact {zCDP} Characterizations for Fundamental Differentially Private Mechanisms},
  year          = {2025},
  month         = oct,
  eprint        = {2510.25746},
  archivePrefix = {arXiv},
  primaryClass  = {cs.CR},
  doi           = {10.48550/arXiv.2510.25746}
}
\endgroup

\end{document}